\documentclass[12pt]{amsart}
\usepackage{color}
\usepackage{amsmath}
\usepackage{amssymb}
\usepackage{graphicx}
\usepackage{pifont}
\usepackage{epsfig}
\usepackage{todonotes}
\usepackage{amscd}
\usepackage{latexsym}
\usepackage{amsfonts}
\usepackage{natbib}
\usepackage{hyperref}
\hypersetup{
    colorlinks=true,
    linkcolor=blue,
    filecolor=magenta,      
    urlcolor=cyan,
    citecolor=blue,
}

\newtheorem{thm}{Theorem}
\newtheorem{lemma}{Lemma}[section]

\newtheorem{defn}{Definition}[section]

\newtheorem{conjecture}{Conjecture}
\renewcommand{\S}{\mathcal{S}}

\newcommand\be{\begin{equation}}
\newcommand\ee{\end{equation}}
\newcommand\bea{\begin{eqnarray}}
\newcommand\eea{\end{eqnarray}}
\newcommand\beaa{\begin{eqnarray*}}
\newcommand\eeaa{\end{eqnarray*}}
\newcommand\bay{\begin{array}}
\newcommand\eay{\end{array}}
\newcommand\ba{\begin{align}}
\newcommand\ea{\end{align}}
\newcommand\bR{{\mathbb{R}}}

\newcommand{\red}[1]{\textcolor{black}{#1}}
\newcommand{\blue}[1]{\textcolor{black}{#1}}

\begin{document}
\bibliographystyle{plainnat}

\title[Evolution of natal dispersal]
{Evolution of natal dispersal in spatially heterogenous environments}
\author[R.S. Cantrell]{Robert Stephen Cantrell}
\address{Department of Mathematics,
University of Miami, Coral Gables, FL 33124}
\email{rsc@math.miami.edu}
\author[C. Cosner]{Chris Cosner}
\address{Department of Mathematics,
University of Miami, Coral Gables, FL 33124}
\email{gcc@math.miami.edu}
\author[Y. Lou]{Yuan Lou}
\address{Institute for Mathematical Sciences, Renmin University of China,
Beijing 100872, PRC
and Department of Mathematics,
Ohio State University, Columbus, OH 43210}
\email{lou@math.ohio-state.edu}
\author[S.J. Schreiber]{Sebastian J. Schreiber}
\address{Department of Evolution and Ecology,
 University of California at Davis, Davis,
 California 95616}
\email{sschreiber@ucdavis.edu}

\date{\today}

\begin{abstract}
 Understanding the evolution of dispersal is an important issue in evolutionary ecology.  For continuous time models in which individuals disperse throughout their lifetime, it has been shown that a balanced dispersal strategy, which results in an ideal free distribution, is evolutionary stable in spatially varying but temporally constant environments. Many species, however, primarily disperse prior to reproduction (natal dispersal) and less commonly between reproductive events (breeding dispersal). These species include territorial species such as birds and reef fish, and sessile species such as plants,  and mollusks. As demographic and dispersal terms combine in a multiplicative way for models of natal dispersal, rather than the additive way for the previously studied models, we develop new mathematical methods to study the evolution of natal dispersal for continuous-time and discrete-time models. A fundamental ecological dichotomy is identified for the non-trivial equilibrium of these models: (i) the per-capita growth rates for individuals in all patches \red{are  equal}   to zero, or (ii) individuals in some patches experience negative per-capita growth rates, while individuals in other patches experience positive per-capita growth rates. The first possibility corresponds to an ideal-free distribution, while the second possibility corresponds to a ``source-sink'' spatial structure. We prove that populations with a dispersal strategy leading to an ideal-free distribution  displace populations  with dispersal strategy leading to a source-sink spatial structure. When there are patches which can not sustain a population, ideal-free strategies can be achieved by sedentary populations, and we show  that these populations can displace populations with any irreducible dispersal strategy. Collectively, these results support that evolution selects for natal or breeding dispersal strategies which lead to ideal-free distributions in spatially heterogenous, but temporally homogenous, environments.

\bigskip

\noindent{\sc Keywords:} Evolution of dispersal; Ideal free distribution;
Evolutionary stability; Patchy environments; {Source-sink populations}
\bigskip

\noindent{\sc AMS Classification}: 34D23, 92D25

\end{abstract}


\maketitle


\section{Introduction}
Dispersal is an important aspect of the life histories of many if not most organisms.  However, it was shown by  \citet{Hastings} that selection generally favors slower rates of dispersal in spatially varying but temporally constant environments.  This is an example of a widespread feature of spatial models in population dynamics and genetics known as the reduction phenomenon, which is that movement or mixing generally reduces growth~\citep{altenberg-12}. Hastings considered types of dispersal such as simple diffusion and symmetric discrete diffusion that did not  allow organisms to perfectly match the distribution  of resources in their environment.  There are dispersal strategies that do allow organisms to match the distribution  of resources in their environment, and it was shown by  \citet{McPeek}  in numerical experiments with discrete time models on two habitat patches  that such strategies were favored by selection. At equilibrium the populations using those strategies had equal fitness in the two patches, which is one of the characteristics of an ideal free distribution.  The ideal free distribution was introduced by \citet{Fretwell} as a {heuristic} theory of how organisms would distribute themselves if individuals could assess their fitness in all locations and were free to move so as to optimize their fitness.  In a population that is at equilibrium and has an ideal free distribution all individuals would have equal fitness and there would be no net movement of individuals, {as a change in local densities would lead to a reduction of fitness for some individuals}.  In the context of population models it is natural to use the per capita growth rate as a proxy for fitness, {in which case the per-capita growth rates of all ideal-free individuals equal zero at equilibrium.} This observation can be used to characterize the ideal free distribution in population models.  It turns out that in many modeling contexts 
\red{ideal free dispersal strategies are evolutionarily stable} in the sense that a population using it  cannot be invaded by an ecologically similar population using a dispersal strategy that does not result in an ideal free distribution. 

In his dissertation, \citet{altenberg-84} conjectured  that strategies leading to an ideal free distribution would be evolutionarily stable.  It turns out \red{ that}  is indeed the case in various types of models.  One approach to modeling the evolution of dispersal, which we shall not pursue here, is based directly on game theory. The implications of the ideal free distribution in that context are described in \citep{vanbaalen-sabelis-93,amnat-00,krivan-03,cressman-etal-04,krivan-etal-08}. The approach that we will take is inspired by the theory of adaptive dynamics.  We will consider {models of populations dispersing in patchy landscapes} and perform what amounts to a pairwise invasibility analysis to compare different dispersal strategies.  {A dispersal strategy is} evolutionarily stable if a population using it can resist invasion by other populations using other strategies.  In fact, we will show that in many cases populations using {dispersal} strategies leading to an ideal free distribution can actually exclude ecologically similar competitors that are using other strategies.  Results on the evolutionary stability of ideal free dispersal have been obtained in various modeling contexts, including  reaction-diffusion-advection equations \citep{cantrell-etal-10, averill-etal-12,korobenko-braverman-14}, discrete diffusion models \citep{pardron-trevisan-06, cantrell-etal-07, cantrell-etal-12},  nonlocal dispersal models~\citep{cosner-14, cantrell-etal-2012b}, and discrete time models \red{ \citep{cantrell-etal-07, KLS}}. {All of these models, however, assume that individuals are either semelparous, as in the case of the discrete-time models, or assume that individuals disperse throughout their lifetime. }

{In many species, dispersing prior to reproducing (natal dispersal) is much more common than dispersing between successive reproductive events (breeding dispersal) \citep{harts-etal-16}. Natal dispersal is the only mode of dispersal for sessile species such as plants with dispersing seeds or sessile marine invertebrates with dispersing larvae. Many territorial species, such as birds or reef fish, often exhibit long natal dispersal distances and little or no dispersal after establishing a territory~\citep{greenwood-harvey-82}. For example, \citet{paradis-etal-98} found that the mean natal dispersal distance for $61$ of $69$ terrestrial bird species was greater than their mean breeding dispersal distance. For species exhibiting significant natal dispersal, the assumption of individuals dispersing throughout their lifetime is inappropriate. A more appropriate simplifying assumption is that these species only disperse a significant amount prior to reproduction and negligible amounts after reproduction. Many continuous time metapopulation models~\citep{hanski-03, Mouquet2002} and discrete time models for iteroparous, as well as semelparous, populations~\citep{hastings-botsford-06,KLS}  have a structure consistent with this assumption: individuals disperse between patches right after birth and settle on one patch for the remainder of their lifetime. 

In the present paper we will derive results on the evolutionary stability of ideal free dispersal strategies for a general class of models {accounting for natal dispersal. We begin by examining the structure of their equilibria and their global stability. The non-trivial equilibrium, when it exists, will be shown to exhibit a dichotomy: per-capita growth rates are equal to zero in all patches (i.e. an ideal free distribution), or some individuals experience negative per-capita growth rates while others experience positive per-capita growth rates. We identify which density-independent dispersal strategies give rise to the ideal-free distributions under equilibrium conditions and show that populations \red{employing} these dispersal strategies exclude populations  \red{employing}  non-ideal free dispersal strategies. In the process, we}  verify a conjecture of  \citet{KLS} and extend some of the results of that paper. Furthermore, we show that in models where {some dispersing individuals} are forced to disperse into {patches only supporting negative per-capita growth rates (sink patches)}, there is selection for slower dispersal.  (In such situations the only strategy that can produce an ideal free distribution is the strategy of no dispersal at all.)

\section{Sources, sinks, and single species dynamics}

\subsection{The general model and assumptions}
We consider two types of models of populations in patchy environments: (i)  models which track population densities in a network of patches, and (ii) patch occupancy models which track the frequencies of occupied sites in a collection of patches \red{,} i.e. metapopulation models.  For both models, we assume that individuals only disperse shortly after reproduction, e.g. plants via seeds, sessile marine invertebrates via larvae, territorial species such as reef fish, etc. For these types of organisms, individuals can experience density- or frequency-dependence in three demographic phases: fecundity (pre-dispersal), settlement (post-dispersal), or survival (adults). Let $u_i(t)$ denote the population density or frequency in patch $i$ at time $t$, where $t\in [0, \infty)$ in continuous time and $t=0,1,...$ in discrete time. Adults living in patch $i$ produce offspring at a rate $f_i(u_i)$ and experience mortality at a rate $m_i(u_i)$. A fraction $d_{ji}$ of offspring disperse from patch $i$ to patch $j$ and only fraction $s_j(u_j)$ of these offspring survive upon arriving in patch $j$.  \red{If}  there are $n$ patches, then the governing equation for $u_i$ is given by
\begin{equation}\label{model-single}
\Delta u_i=s_i(u_i) \sum_{j=1}^n d_{ij} f_j(u_j) u_j
-m_i(u_i) u_i,
\quad 1\le i\le n,
\end{equation}
where $\Delta u_i =\frac{du_i}{dt}$  in continuous time and $\Delta u_i = u'_i-u_i$ in discrete time;
$u'_i$ denotes the population density in patch $i$ in the next time step,
i.e. $u'_i(t)=u_i(t+1)$. Denote $\bR_{+}=[0, \infty)$. {Let $\S=\bR_+$ for the population density models and $\S=[0,1]$ for the population frequency models. Then $\S^n$ is the state space for the models.}  We make the following assumptions.

\medskip

\noindent{\bf(A1)} Matrix $(d_{ij})$ is non-negative, column stochastic, and irreducible in the continuous time case and primitive in the discrete time case. \red{A square matrix is irreducible if it is not similar via a permutation to a block upper triangular matrix.
 A primitive matrix is a square nonnegative matrix some power of which is positive}. Biologically, no individuals are lost while dispersing and after enough generations, the descendants of an adult from patch $i$ can be found in all patches. 		

\medskip
\noindent{\bf(A2)} $f_i, s_i: \S \to \bR_{+}$ are continuous, positive, 
\red{non-increasing} functions, and $m_i: \S\to \bR_{+}$ is continuous, positive, \red{non-decreasing}. Biologically, reproduction and survival \red{rates} decrease with population density, while \red{mortality} increases with density.

\medskip
\noindent{\bf (A3)}  $u_if_i(u_i): \S_{+}\to \bR_{+}$ is strictly increasing. Biologically, as the population gets larger in patch $i$, the more offspring are produced by the population in patch $i$. In the discrete-time case, we require the stronger hypothesis that
\[
\frac{\partial}{\partial u_i}\left(s_i(u_i) d_{ii}f_i(u_i)+(1-m_i(u_i))u_i\right) >0 \mbox{ for } u_i \in \S, i=1,2,\dots, n.
\]
This stronger hypothesis is needed to ensure monotonicity of the discrete time population dynamics.
\medskip

For each $i$ and population density $u_i$, define
\begin{equation}\label{eq:g-i}
g_i(u_i):=\frac{s_i(u_i) f_i(u_i)}{m_i(u_i)}.
\end{equation}
If the population density were held constant at $u_i$, then $g_i(u_i)$  equals the mean number of surviving offspring produced by an individual remaining in patch $i$ during its life time. Namely, $g_i(u_i)$ is the reproductive number of individuals living in patch $i$ with the fixed local density $u_i$. Hence, we view $g_i(u_i)$ as the fitness of an individual remaining in patch $i$. By (A2), $g_i$ are continuous, positive, decreasing functions. We make the following stronger  assumption on $g_i$:

\medskip
\noindent{\bf(A4)} $g_i$ is strictly decreasing on $\S$ and $\lim_{u_i\to\infty} g_i(u_i)<1$ for the population density models and $g_i(1)<1$ for the population frequency models. Biologically, fitness within a patch decreases with density and, at high enough densities, individual fitness is less than one i.e. individuals don't replace themselves.
\medskip

\noindent{} The term $g_i(0)$ determines whether individuals remaining in patch $i$ replace themselves during their lifetime under low-density conditions. If $g_i(0)>1$, individuals do replace themselves and patch $i$ is called a \emph{source patch}. In the absence of dispersal, source patches are able to sustain a persisting population. If $g_i(0)\le 1$, individuals at best just replace themselves and patch $i$ is called \emph{sink patch}. In the absence of dispersal, sink populations go asymptotically extinct.

We describe two classes of models in the literature which satisfy our assumptions.

\subsubsection{The Mouquet-Loreau model} A particular example of the population frequency models is due to  \citet{Mouquet2002,mouquet-loreau-03}. They introduced a model of metapopulations incorporating spatial structure at two scales: the within-patch scale and the between-patch scale. At the within-patch scale, the environment consists of a collection of identical sites, each of which can be occupied at most by one individual. For this population frequency model, the fraction of sites occupied in patch $i$ is $u_i$ and $\S=[0,1]$. Individuals within patch $i$ produce offspring at rate $b_i$ of which a fraction $d_{ji}$ disperse to patch $j$. Offspring arriving in site $i$ only survive if they colonize an empty site. Assuming each offspring colonizes sites at random, the probability of colonizing an empty site in patch $i$ is $1-u_i$. Offspring attempting to colonize an occupied site die.  Adults in site $i$ experience a per-capita mortality rate of $m_i$. Thus, the dynamics in patch $i$ are
\begin{equation}\label{eq:1}
\frac{d u_i}{dt}
=(1-u_i)  \sum_{j=1}^n d_{ij} b_j u_j -m_i u_i \quad i=1,2,\dots,n.
\end{equation}
In terms of \eqref{model-single}, we have $s_i(u_i)=1-u_i$, $f_i(u_i)=b_i $, $m_i(u_i)=m_i$, and $g_i(u_i)=b_i (1-u_i)/m_i$. For this model, a patch is a source if $b_i>m_i$, and a sink otherwise.

\subsubsection{The Kirkland et al. model} A particular example of the population density models is due to \citet{KLS}. They described a metapopulation model of semelparous populations (e.g. monocarpic plants, most insect species, certain fish species) which compete for \red{limited}  resources, reproduce, and  disperse as offspring. For these discrete-time models, all reproducing adults die and, consequently, $m_i(u_i)=1$. The number of surviving offspring produced per individual in patch $i$ is given by a decreasing function $f_i(u_i)$ satisfying \textbf{A3} and $\lim_{u_i\to\infty} f_i(u_i)<1$. A standard choice of $f_i$ is the \citet{beverton-holt-57} model $f_i(u_i)=\lambda_i/(1+a_iu_i)$ where $\lambda_i$ is the number of surviving offspring produced by an individual at low population densities, and $a_i$ measures the strength of intraspecific competition. As all offspring are able to colonize a patch, this model assumes $s_i(u_i)=1$ and, consequently, $g_i(u_i)=f_i(u_i)$. Source patches in this model satisfy $\lambda_i>1$, while $\lambda_i\le 1$ for sink patches.

\subsection{Global stability, ideal-free distributions, and source-sink landscapes}

If assumptions (A1) and (A2) hold, any non-zero equilibrium
of \eqref{model-single}  must be
positive in each component. Therefore, under these assumptions, \eqref{model-single} only has two types of equilibria: the zero equilibrium $u^*=(0,0,\dots,0)$ and positive equilibria $u^*=(u_1^*,u_2^*,\dots,u_n^*)$ where $u_i^*>0$ for all $i$.
\red{See Lemma \ref{lemma:single2-1} below}.

To determine whether the populations persist or not, we linearize \eqref{model-single} at the zero equilibrium:
\begin{equation}\label{linearize}
\Delta u = L u \mbox{ where } L=(L_{ij}) \mbox{ with } L_{ij} = s_i(0)d_{ij}f_j(0)-m_i(0)\red{\delta_{ij}}.
\end{equation}
Let $\rho(L)$ denote the \emph{stability modulus} of the matrix $L$ i.e. the largest of the real parts of the eigenvalues of $L$. The following lemma implies that  persistence or extinction of the population in \eqref{model-single} is determined by the sign of $\rho(L)$.

\begin{lemma}\label{lemma:single2-1} Suppose that
{\rm (A1)-(A4)} hold and let $L$ be given by \eqref{linearize}. If $\rho(L)>0$,  then \eqref{model-single} has a unique positive equilibrium, which
is globally asymptotically stable among
non-negative and not identically zero initial data {in $\S^n$}. If $\rho(L)\le 0$, then the zero equilibrium $u^*=0$ is globally stable among all non-negative initial data {in $\S^n$}.
\end{lemma}
\noindent{}The proof of Lemma \ref{lemma:single2-1} for discrete time is similar to that of Theorem 2.1
in \cite{KLS}.
The proof of Lemma \ref{lemma:single2-1} for continuous time
 is based upon the following facts: (i) system \eqref{model-single}
is strongly monotone and sublinear [(A1)-(A3)];
(ii) stability of the zero equilibrium (0, ..., 0) of \eqref{model-single}
is determined by $\rho(L)$; (iii) all positive solution
trajectories of system \eqref{model-single} are bounded [(A1), (A2), (A4)].
We omit the proofs which follow from standard monotone dynamical systems theory~\citep{Smith}.

Whenever there is a positive equilibrium for \eqref{model-single}, the following theorem proves the existence of an ecological dichotomy. \\

\begin{thm}~\label{source-sink} Suppose that {\rm (A1)-(A4)} hold and $u^*$ is a positive equilibrium for \eqref{model-single}. Then either (i)  $g_i(u_i^*)=1$ for all $i$, or (ii)  there exist $i$ and $j$ such that $g_i(u_i^*)<1$ and $g_j(u_j^*)>1$.
\end{thm}

In this dichotomy, either the fitnesses of individuals equal one in all patches, or there are patches where the fitness is greater than one, and other patches where it is less than one. In the first case, individuals can not increase their fitness by moving to any other patch and exhibit an \emph{ideal free distribution}~\citep{Fretwell}. In the second case, fitness is greater than one in some patches and less than one in other patches. As $g_i$ are decreasing functions, the patches where $g_i(u_i^*)>1$ are sources. The patches where $g_i(u_i^*)<1$ can be either sinks or sources. In the latter case, they are known as ``pseudo-sinks'' as they appear to be sinks, but populations in these patches can persist in the absence of immigration~\citep{watkinson-sutherland-95}. Pseudo-sinks arise when immigration into a source patch increases the equilibrium population density beyond the patch's carrying capacity. As the fitness functions are decreasing, only the second option is possible whenever there are ``true'' sinks in the landscape (i.e. $g_i(0)<1$ for some $i$).

To prove Theorem~\ref{source-sink} we need the following \blue{two} lemmas which also \blue{are} useful for our later theorems. \blue{First, we need a property of the column sum norm of a matrix: $\| A\|_\infty=\mbox{max}\{\sum_i |A_{ij}|: 1\le j \le n\}$. As $\|A\|_\infty=\max_{\|u\|_\infty=1}\|u^T A\|_\infty$ (where $u^T$ denotes the transpose of $u$), $\|A\|_\infty$ corresponds to an operator norm of $A$ with respect to the $\ell_\infty$ norm $\|u\|_\infty=\max_i|u_i|$.}

\blue{
\begin{lemma}\label{lemma:upper}
Let $A$ be a non-negative, irreducible matrix. Then 
\[
\rho(A)\le \| A\|_\infty
\]
where equality implies that all of the column sums of $A$ equal $\rho(A)$.
\end{lemma}
}

\blue{\begin{proof} Let $u$ be a left eigenvector of $A$ associated with the eigenvalue $\rho(A)$. By the Perron-Frobenius Theorem, $u$ can be choosen to have strictly positive entries with $\|u\|_\infty=1$. Let $j$ be such that $u_j=1$. Then 
\begin{equation}\label{eq:upper}
\rho(A)= \sum_i A_{ij} u_i \le \sum_i A_{ij} \le \|A\|_\infty
\end{equation}
which gives the first assertion of the lemma.
Now suppose that $\rho(A)=\|A\|_\infty$. Define $B=A+I$ where $I$ denotes the identity matrix. As $B$ is primitive, there exists $n\ge 1$ such that $B^n$ has strictly positive entries. Then $\rho(B)=\rho(A)+1=\|A\|_\infty+1=\|B\|_\infty$ and
\[
\rho(B^n)=\rho(B)^n = (\|B\|_\infty)^n \ge \|B^n\|_\infty \ge \rho(B^n)
\]
where the first inequality follows from $\|\cdot\|_\infty$ being an operator norm and the second inequality follows from  \eqref{eq:upper} applied to $B^n$. Hence, $\rho(B^n)=\|B^n\|_\infty$. Since $B^n$ has strictly positive entries, if $u_i <1$ for some $i$, then 
\[
\rho(B^n)= \sum_i (B^n)_{ij} u_i < \sum_i (B^n)_{ij} \le \|B^n\|_\infty
\]
 contradicting that $\rho(B^n)=\|B^n\|_{\infty}$. Therefore, $u=(1,1,\dots,1)$. As $uA=\rho(A)u$, all of the column sums of $A$ equal $\rho(A)$.
\end{proof}}

\blue{We use Lemma~\ref{lemma:upper} to prove the following crucial inequality. }

\begin{lemma}~\label{reduction}
Let $D$ be a column stochastic, non-negative, irreducible matrix and $S,\Lambda, M$ be diagonal matrices with positive diagonal entries. Then
\[
\rho(S D \Lambda -M) \le \rho(S\Lambda-M)
\]
with equality if and only if $S\Lambda-M$ is a scalar matrix, 
\red{i.e. a diagonal matrix with all its main diagonal entries equal}.
\end{lemma}

If $D$ corresponds to a dispersal matrix, and $\Lambda$, $S$, and $M$ have diagonal elements $s_i(u_i)$, $f_i(u_i)$, and  $m_i(u_i)$, then this lemma implies that dispersal decreases the ``metapopulation'' growth rate $\rho(SD\Lambda -M)$.  This is a type of reduction phenomena  that is similar to those described by \citet{altenberg-12}.

\begin{proof}  Consider
\[
S^{-1} \left( S D \Lambda - M + \rho(M) I\right) S= D \underbrace{\Lambda S}_{=A}+ \underbrace{\rho(M) I - M}_{=B}
\]
whose stability modulus equals $\rho(S D \Lambda -M)+\rho(M)$. By definition, $A$ is a diagonal matrix with strictly positive diagonal entries $a_i$ and $B$ is a diagonal matrix with non-negative entries $b_i$.

Define the matrix $C$  by $C_{ij} = D_{ij} \frac{a_j}{a_j+b_j}$ for $i\neq j$ and $C_{ii}=1-\sum_j C_{ji}$. $C$ is a column stochastic, irreducible, non-negative matrix. By definition of $C$, we have $C(A+B)=DA+B$.
\blue{Lemma~\ref{lemma:upper} implies that
\[
\rho(C (A+B)) \le \| C(A+B)\|_\infty  =\rho(A+B)
\]
with a strict inequality whenever $A+B$ is non-scalar (as then the column sums of $C(A+B)$ are not all equal).} Hence,
\[
\rho(SD\Lambda -M) = \rho(C (A+B)) - \rho(M)  \le \rho(A+B) - \rho(M)  = \rho(\Lambda S -M)
\]
with a strict inequality whenever $A+B$ is non-scalar.
\end{proof}

\noindent{\emph{ Proof of Theorem \ref{source-sink}.}} Let $u^*$ be a positive equilibrium. Lemma~\ref{reduction} with $D=(d_{ij})$, $S$ with diagonal elements $s_i(u_i^*)$, $\Lambda$ with diagonal elements $f_i(u_i^*)$, and $M$ with diagonal elements $m_i(u_i^*)$ implies
\begin{eqnarray}\label{useful}
\nonumber\max_i \left\{ s_i(u_i^*)f_i(u_i^*)-m_i(u_i^*)\right\}&=& \rho (S\Lambda -M)\\
&\ge& \rho(SD\Lambda-M) \\
\nonumber&=&0
\end{eqnarray}
where the last line follows from $u^*$ being a positive equilibrium, and where the inequality in the second line is an equality if and only if $S\Lambda-M$ is a scalar matrix.

Suppose $S\Lambda-M$ is a scalar matrix. Then \eqref{useful} implies $s_i(\red{u_i^*})f_i(u_i^*)-m_i(u_i^*)=0$ for all $i$. Equivalently, $g_i(u_i^*)=1$ for all $i$.

Suppose $S\Lambda-M$ is not a scalar matrix. Then inequality \eqref{useful} is strict and there exists $j$ such that $s_j(u_j^*)f_j(u_j^*)-m_j(u_j^*)>0$. Equivalently, $g_j(u_j^*)>1$. Furthermore, as $u^*$ is an equilibrium,
\begin{eqnarray*}
0&=&\sum_i \Delta u_i/s_i(u_i)\Big|_{u=u^*}\\
&=& \sum_{i,\ell} f_\ell(u_\ell^*)d_{i\ell} u_\ell^* - \sum_i m_i(u_i^*)u_i^*/s_i(u_i^*)\\
&=& \sum_i \left( f_i(u_i^*)-m_i(u_i^*)/s_i(u_i^*)\right)u_i^*.
\end{eqnarray*}
As $f_j(u_j^*)-m_j(u_j^*)/s_j(u_j^*)$ in this final sum is positive, there must be an $i$ such that
$f_i(u_i^*)-m_j(u_i^*)/s_i(u_i^*)<0$  i.e. $g_i(u_i^*)<1$.
\qed

\section{Evolution of ideal-free distributions}

To study the evolution of dispersal, we consider two populations that
have the same population dynamics and  differ only in their dispersal strategies.
Let $u_i$ and $v_i$ be the densities of two competing species in patch $i$.
The governing equations for $u_i$ and $v_i$ are
\begin{equation}\label{model-2species}
\aligned
&\Delta u_i=s_i(u_i+v_i) \sum_{j=1}^n d_{ij} f_j(u_j+v_j) u_j
-m_i(u_i+v_i) u_i,
\\
&\Delta v_i=s_i(u_i+v_i) \sum_{j=1}^n D_{ij} f_j(u_j+v_j) v_j
-m_i(u_i+v_i) v_i,
\endaligned
\end{equation}
where the matrix $(D_{ij})_{ij}$ is a non-negative, column stochastic matrix. {The state space for \eqref{model-2species} is $\S^{2n}$.}

We are interested in identifying classes of dispersal matrices $(d_{ij})$ such that populations with this dispersal strategy displace populations with a different dispersal strategy. Consistent with classical ecological theory, Theorem~\ref{source-sink} suggests a solution. Namely, assuming $(d_{ij})$ is irreducible, let $u^*$ be a positive equilibrium of the single species model \eqref{model-single}. There are two possibilities: there are patches with $g_i(u_i^*)>1$, or there is an ideal free distribution with $g_i(u_i^*)=1$ for all $i$. Intuitively, for the first option, another population where individuals prefer going to the source patch (i.e. $d_{ij}\approx 1$ for all $j$ where $i$ is such that $g_i(u_i^*)>1$) can invade the $(u^*,0)$ equilibrium. Hence, the only candidate\red{s} for the ``unbeatable'' strategy are those which lead to an ideal-free distribution.

As equilibrium ideal-free distributions with an irreducible $(d_{ij})$ only occur if all patches are sources, we begin by considering this case. Then we  consider the case in which there are sink patches.

\subsection{Landscapes of source patches}
Assume there are only source patches i.e. $g_i(0)>1$ for all $i$. For these landscapes, for each $i$, there is a unique $u_i^*>0$, the \emph{carrying capacity} of patch $i$,  such that $g_i(u_i^*)=1$. Let  $u^*=(u_1^*,...,u^*_n)$ which corresponds to an ideal free distribution. To understand which dispersal matrices yield an ideal free distribution in the single species model \eqref{model-single}, we  have the following lemma.

\begin{lemma} Assume $g_i(0)>1$ for all $i$. If $g_i(u_i^*)=1$ for all $i$, then $u^*$ is an equilibrium of \eqref{model-single}
if and only if
\begin{equation}\label{eq:IFD}
\sum_{j=1}^n d_{ij} f_j(u_j^*)u_j^*=f_i(u_i^*) u_i^*
\end{equation}
for all $i$.
\end{lemma}
\begin{proof}
$u^*$ is an equilibrium of \eqref{model-single} if and only if
$g_i(u_i^*)\sum_{j=1}^n d_{ij} f_j(u_j^*)u_j^*=f_i(u_i^*) u_i^*$
for every $i$, which is equivalent to \eqref{eq:IFD} since
$g_i(u_i^*)=1$ for every $i$.
\end{proof}

As the matrix $(d_{ij})$ is column stochastic, equation \eqref{eq:IFD} implies the flux of individuals coming into a patch equal the flux of individuals leaving a patch:
\[
\sum_{j=1}^n d_{ij} f_j(u_j^*)u_j^*=\sum_{j=1}^n d_{ji} f_i(u_i^*) u_i^*.
\]
Hence, an ideal free distribution implies that there is no net loss or gain
of individuals from dispersal, so there are no costs or
benefits arising directly from dispersal. \citet{McPeek} call this ``balanced dispersal''. 	

\begin{defn}  
We say that $(d_{ij})$ is an ideal free dispersal strategy
with respect to $u^*$ if \eqref{eq:IFD} holds
for every $i$, and that  $(d_{ij})$ is not an ideal free dispersal strategy
if \eqref{eq:IFD} fails for some $i$.
\end{defn}

There exists
 a close connection between
an ideal free dispersal strategy and a \emph{line-sum-symmetric matrix}:
a square matrix $A$ such that
 the sum of the elements in the i-th row of $A$ equals the sum of
the elements in the i-th column of $A$ for every $i$.

\begin{lemma}\label{lemma:2-1} If matrix $(d_{ij})$ is an ideal free dispersal strategy with
respect to $u^*$, then matrix $(d_{ij} f_j(u_j^*)u_j^*)$ is line-sum-symmetric.
\end{lemma}
\begin{proof}
For each $i$,
$$
\sum_{j} d_{ij} f_j(u_j^*)u_j^*=f_i(u_i^*) u_i^*
=\left(\sum_{j} d_{ji}\right) f_i(u_i^*)u_i^*
=\sum_{j} \left(d_{ji}f_i(u_i^*)u_i^* \right),
$$
i.e., matrix $(d_{ij} f_j(u_j^*)u_j^*)$ is line-sum-symmetric.
\end{proof}

The following classification result for
line-sum-symmetric matrices is given in \cite{eaves-etal-85} and
plays important role in the proof of our first main result.
\begin{thm}\label{thm:ls}
Let $A=(a_{ij})$ be an $n\times n$ nonnegative matrix. Then $A$ is line-sum-symmetric if
and only if
\begin{equation}\label{eq:linesymmetric-1}
\sum_{i, j=1}^n a_{ij} \frac{x_i}{x_j}\ge \sum_{i, j=1}^n a_{ij}
\end{equation}
for all $x_i>0$, $1\le i\le n$. Moreover, if $A$ is irreducible and line-sum-symmetric, equality in
{\rm(\ref{eq:linesymmetric-1})}
holds if and only if all the coordinates of $x=(x_1, ..., x_n)$ coincide, i.e., $x_i=x_j$ for any
$1\le i, j\le n$.
\end{thm}

Our main result for sink-free landscapes is that ideal-free dispersal strategies always outcompete non-ideal-free strategies. This result provides a positive answer to Conjecture 5.1 of \citet{KLS}  for the Kirkland et al. model with $f_i(u_i)=\lambda_i/(1+a_i u_i)$. The proof is given in section~\ref{proof:main:1}.

\begin{thm}\label{main:1}
 Suppose that {\rm(A1)-(A4)} hold, the matrix $(D_{ij})$ satisfies {\rm(A1)}, and $g_i(0)>1$ for all $i$. If  $(d_{ij})$ is an ideal free dispersal strategy with respect to $u^*$
 but $(D_{ij})$ is not, then $(u^*, 0)$ is  globally asymptotically stable among
non-negative and not identically zero initial data for \eqref{model-2species}.
\end{thm}

Theorem~\ref{main:1} suggests that evolution selects for these ideal-free strategies. But what happens when different ideal-free strategies compete? The following result implies there are neutral dynamics with neither strategy outcompeting the other. Hence, one would expect that genetic drift, which is not included in these deterministic models, to play an important role. { To state the theorem, recall the $\omega$-limit set, $\omega(K)$, of a compact set  $K\subset \S^{2n}$ is given by the closure of $\cap_{T\ge 1} \cup_{t\ge T} \{ (u(t),v(t)): (u(0),v(0))\in K\}$. A compact set  $A\subset \S^{2n}$ is an \emph{attractor} if there exists a compact neighborhood $K$ such that $\omega(K)=A$. The \emph{basin of attraction} of $A$ is the set of points $(u,v)\in \S^{2n}$ such that $\omega(\{(u,v)\})\subset A$. }

\begin{thm}\label{main:1B}
 Suppose that {\rm(A1)-(A4)} hold, the matrix $(D_{ij})$ satisfies {\rm(A1)}, and $g_i(0)>1$ for all $i$. If  $(d_{ij})$ and $(D_{ij})$ are ideal free dispersal strategies \red{with respect to $u^*$,
 then} the set of non-negative equilibria consists of $(0,0)$ and the line of equilibria given by
 \[
 \mathcal{E}=\left\{(\alpha u^*,(1-\alpha)u^*): 0\le \alpha \le 1\right\}.
 \]
{Furthermore, if $f_i$, $s_i$, $m_i$ are continuously differentiable and satisfy} { $f_i'(x)+s_i'(x)-m_i'(x)<0$} {for all $x\in \bR_+$ and $1\le i\le n$, and, in addition, for the discrete-time case \[\frac{\partial}{\partial x} \left( s_i(x+y)d_{ii}f_i(x+y)x+(1-m_i(x+y))x\right)>0\] for all $x,y\in \bR_+$ and $1\le i\le n$, then $\mathcal {E}$ is an attractor and the $\omega$-limit set for every point in its basin of attraction consists of a single point in $\mathcal{E}$. }
\end{thm}

Theorem~\ref{main:1B} is a generalization of \citet[Proposition 5.3]{KLS}. {The proof is given in section~\ref{proof:main:1B}.}
 {The additional assumptions to ensure $\mathcal{E}$ is an attractor correspond to a strengthening of our standing assumptions \textbf{(A2)} and \textbf{(A3)}.}

\subsection{Source-sink landscapes} Now, let us consider landscapes with sink as well as source patches. Specifically, assume for some $1\le k<n$,  $g_i(0)\le 1$ for $k+1\le i\le n$, and $g_i(0)>1$ for $1\le i\le k$. For the source patches $1\le i\le k$, there exist unique $u_i^*>0$ such that $g_i(u_i^*)=1$. Then $u^*=(u_1^*,u_2^*,\dots, u_k^*,0,\dots,0)$ corresponds to an ideal-free distribution as individuals (which only exist in the source patches) can not increase their fitness by moving to any other patch.  As $u_i^*=0$ for some $i$, the only way this distribution can be realized as an equilibrium of \eqref{model-single} is if $(d_{ij})$ is reducible. For example, if $(d_{ij})$ is the identity matrix i.e. the populations are sedentary, then $u^*$ is a globally stable equilibrium of \eqref{model-single} for non-negative initial data whose first $k$ coordinates are positive.

For a sedentary population competing with a dispersing population, the governing dynamics are given by
\begin{equation}\label{model-2species-v2}
\aligned
&\Delta u_i=s_i(u_i+v_i) f_i(u_i+v_i) u_i
-m_i(u_i+v_i) u_i,
\\
&\Delta v_i=s_i(u_i+v_i) \sum_{j=1}^n D_{ij} f_j(u_j+v_j) v_j
-m_i(u_i+v_i) v_i.
\endaligned
\end{equation}

The following theorem shows that these sedentary populations always outcompete populations with an irreducible dispersal matrix. Intuitively, irreducibility of  $(D_{ij})$ implies there is dispersal into sink habitats, resulting in a loss of individuals.

\begin{thm}\label{main:2}
Suppose that {\rm(A2)-(A4)} hold, the matrix $(D_{ij})$ satisfies {\rm(A1)}, and there exists a $k<n$ such that $g_i(0)>1$ if and only if $1\le i\le k$ {and $g_i(0)<1$ for some $i>k$}. Let $u_i^*>0$ be the unique solution to $g_i(u_i)=1$ for $1\le i\le k$ and $u_i^*=0$ for $i>k$. Then $(u^*, 0)$ is  globally asymptotically stable among positive initial data for \eqref{model-2species-v2}.
\end{thm}

A proof of Theorem~\ref{main:2} is given in section~\ref{proof:main:2}. We conjecture that there is a more general class of reducible dispersal strategies that lead to exclusion of all irreducible dispersal strategies. Specifically, if $(d_{ij})$ is a matrix which (i) restricted to patches $1\le i\le k$ is an ideal-free dispersal strategy, and (ii) $d_{ij}=0$ whenever $i>k$ or $j>k$.

\begin{conjecture}\label{main:2b}
Suppose that {\rm(A2)-(A4)} hold, the matrix $(d_{ij})$ is ideal-free in the above sense, and the matrix $(D_{ij})$ satisfies {\rm(A1)}. Let $u_i^*>0$ be the unique solution to $g_i(u_i)=1$ for $1\le i\le k$ and $u_i^*=0$ for $i>k$. Then $(u^*, 0)$ is  globally asymptotically stable among positive initial data for \eqref{model-2species}.
\end{conjecture}

\section{Conclusions}

{For populations dispersing prior to reproduction in a patchy landscape, we have found a fundamental dichotomy about their equilibrium state: either the populations have per-capita growth rates equal to zero in all occupied patches, or their per-capita growth rates are positive in some patches and negative in others. The first possibility corresponds to an ideal-free distribution in the sense of \citet{Fretwell} as 
\red{individual fitnesses}  are equal in all occupied patches. Under these ideal-free equilibrium conditions, the populations also exhibit ``balanced dispersal'' as the immigration rate 
 \red{into each patch} 
  is balanced by immigration out of the patch. The second possibility in the dichotomy corresponds to a landscape, in the sense of \citet{pulliam-88}, containing source and sink patches (where sink patches may be pseudo-sinks in the parlance of \citet{watkinson-sutherland-95}). In source patches, births exceed deaths and emigration exceeds immigration, while in sink patches the opposite occurs. While this dichotomy has been the focus of a series empirical papers~\citep{doncaster-etal-97,diffendorfer-98,tattersall-04}, our Theorem~\ref{source-sink} is, to the best of our knowledge, the first mathematical demonstration of this dichotomy. While our proof was for models with natal dispersal, the proof should apply to many more model types including those with breeding dispersal as well as natal dispersal.}

{Given this ecological dichotomy, one can ask which one is favored by natural selection acting on natal dispersal.} For patchy landscapes, we have shown that {natal} dispersal strategies leading to an ideal free distribution are evolutionarily stable.  If all patches can support local populations, we observe that the strategies leading to an ideal free distribution can be characterized in terms of the line sum symmetry of certain matrices  constructed from the dispersal matrices and the carrying capacities of the patches in the models.  In the case where some patches {can not support local populations} and the dispersal matrix is irreducible we observe that the only dispersal strategy that can produce an ideal free distribution is the strategy of no movement at all.  These results, together with similar results already known for {models of semelparous populations or populations that continually disperse  throughout their lifetime}, support the conclusion that in spatially varying but temporally constant environments the dispersal strategies that are evolutionarily stable are those that lead to an ideal free distribution of the population using them. An underlying biological reason for this is that such strategies allow populations to perfectly match the levels of resources across their environment.

{Our analysis extends the reduction principle, in which movement or mixing generally reduces growth~\citep{altenberg-12}, to populations with natal dispersal. However, beyond the ``natal versus breeding dispersal'' dichotomy~\citep{greenwood-harvey-82}, populations are structured by other states such as size, age, gender, or stages, and individuals in different states may have different dispersal propensities~\citep{harts-etal-16}. \citet{laa-06}'s analysis of discrete-time linear models demonstrates that the reduction principle may not hold whenever two or more stages are dispersing. More precisely, they proved whenever there are cycles in the population's dispersal graph that involve multiple stages, there exists a matrix model for which this form of movement increases the population growth rate. Conversely, when no such cycles exist, their analysis suggests that the reduction principle might hold. This raises an interesting future challenge: what are evolutionarily stable strategies for dispersal for stage-structured populations? In particular, we conjecture that if only one stage disperses, dispersal strategies of this stage leading to an ideal-free distribution are evolutionarily stable. }

\medskip

\noindent{\bf Acknowledgements.}
We would like to thank the editor and two anonymous referees for their helpful comments which improve the manuscript.
This research was partially supported by the NSF grants DMS-0816068, DMS-1118623, and DMS-1514752 (RSC, CC),
DMS-1411476 (YL), DMS-1022639 and DMS-1313418 (SJS)
and by National Natural Science Foundation of China grants No. 11571363 and No. 11571364 (YL).

  \bibliography{CCLS}

\section{Proof of Theorem \ref{main:1}\label{proof:main:1}}
One underlying mathematical difficulty in the proof of Theorem \ref{main:1}
 is that  $(u^*, 0)$ is neutrally stable, so even determining the local stability
of $(u^*, 0)$ seems to be of interest. Our main idea is to
establish the following lemma  and apply the theory of strongly monotone dynamical systems.

If assumptions (A1)-(A4) hold and $(d_{ij})$ is an ideal free dispersal strategy,
\eqref{model-2species} has three special equilibria: the trivial equilibrium $(0, 0)$,
and two semi-trivial equilibria $(u^*, 0)$ and $(0, v^*)$ assuming the latter exists. The following
result ensures that these are all possible non-negative equilibria of \eqref{model-2species}.
Note that if $(u, v)$ is a non-negative \red{equilibrium} of \eqref{model-2species}
and $(u, v)\not=(0, 0), (u^*,0), (0, v^*)$, then by (A1) and (A2), all components of $u, v$ are positive,
i.e., $(u, v)$ is a positive equilibrium.

\begin{lemma}\label{lemma:c-1} Suppose that {\rm(A1)-(A4) hold}. If
 $(d_{ij})$ is an ideal free dispersal strategy with respect with $u^*$
but  $(D_{ij})$ is not, then \eqref{model-2species}
has no positive equilibrium.
\end{lemma}
\begin{proof}
Suppose that $(u, v)$ is a positive equilibrium of \eqref{model-2species}.
Then
\begin{equation}\label{model-2species-1}
\aligned
&s_i(u_i+v_i) \sum_{j=1}^n d_{ij} f_j(u_j+v_j) u_j
-m_i(u_i+v_i) u_i=0,
\\
&s_i(u_i+v_i) \sum_{j=1}^n D_{ij} f_j(u_j+v_j) v_j
-m_i(v_i+v_i) v_i=0
\endaligned
\end{equation}
for every $i$.
Dividing the first equation of
\eqref{model-2species-1} by $s_i(u_i+v_i)$ and summing up in $i$ we obtain
\begin{equation}\label{model-2species-2}
\aligned
\sum_{i=1}^n
\frac{m_i(u_i+v_i)}{s_i(u_i+v_i)} u_i
&=\sum_{i, j=1}^n d_{ij} f_j(u_j+v_j) u_j
\\
&=\sum_{j=1}^n  f_j(u_j+v_j) u_j \left(\sum_{i=1}^n d_{ij}\right)
\\
&=\sum_{j=1}^n  f_j(u_j+v_j) u_j.
\endaligned
\end{equation}

Similarly,
\begin{equation}\label{model-2species-3}
\sum_{i=1}^n
\frac{m_i(u_i+v_i)}{s_i(u_i+v_i)} v_i
=\sum_{j=1}^n  f_j(u_j+v_j) v_j.
\end{equation}

By \eqref{model-2species-2} and \eqref{model-2species-3},
we  obtain
\begin{equation}\label{model-2species-4}
\sum_{i=1}^n
\frac{m_i(u_i+v_i)}{s_i(u_i+v_i)} (u_i+v_i)
=\sum_{j=1}^n  f_j(u_j+v_j) (u_j+v_j),
\end{equation}
which can be rewritten as
\begin{equation}\label{model-2species-5}
\sum_{i=1}^n
(u_i+v_i) f_i(u_i+v_i)
\left[g^{-1}_i(u_i+v_i)-1\right]=0.
\end{equation}
{(Here $g_i^{-1}$ denotes the reciprocal of $g_i$; not the inverse function.)}

Dividing the first equation of
\eqref{model-2species-1} by $s_i(u_i+v_i)u_i f_i(u_i+v_i)/[f_i(u_i^*) u_i^*]$ and summing up in $i$ we obtain
\begin{equation}\label{model-2species-6}
\sum_{i,j}d_{ij} \frac{f_j(u_j+v_j)u_j f_i(u_i^*)u_i^*}
{f_i(u_i+v_i) u_i}
=\sum_{i=1}^n \frac{u_i^* f_i(u_i^*) m_i(u_i+v_i)}{s_i(u_i+v_i) f_i(u_i+v_i)}.
\end{equation}

By Lemma \ref{lemma:2-1}, matrix $(d_{ij} f_j(u_j^*)u_j^*)$ is line-sum-symmetric.
Setting $a_{ij}=d_{ij} f_j(u_j^*)u_j^*$ and
$$
x_i=\frac{f_i(u_i^*)u_i^*}{f_i(u_i+v_i)u_i}
$$
in Theorem \ref{thm:ls}, we obtain
\begin{equation}\label{model-2species-6b}
\sum_{i,j}d_{ij}f_j(u_j^*) u_j^* \frac{f_i(u_i^*)u_i^*/[f_i(u_i+v_i)u_i]}
{f_j(u_j^*)u_j^*/[f_j(u_j+v_j)u_j]}
\ge \sum_{i, j} d_{ji} f_i(u_i^*) u_i^*,
\end{equation}
which can be simplified as
\begin{equation}\label{model-2species-7}
\sum_{i,j}d_{ij} \frac{f_i(u_i^*)u_i^* f_j(u_j+v_j)u_j}
{f_i(u_i+v_i)u_i}
\ge \sum_{i} f_i(u_i^*) u_i^*.
\end{equation}

By \eqref{model-2species-6} and \eqref{model-2species-7} we have
\begin{equation}
\sum_{i=1}^n \frac{u_i^* f_i(u_i^*) m_i(u_i+v_i)}{s_i(u_i+v_i) f_i(u_i+v_i)}
\ge \sum_{i} f_i(u_i^*) u_i^*,
\end{equation}
which can be written as
\begin{equation}\label{model-2species-8}
\sum_{i=1}^n u_i^* f_i(u_i^*) \left[g^{-1}_i(u_i+v_i)-1\right]\ge 0.
\end{equation}

It follows from \eqref{model-2species-5}, \eqref{model-2species-8} and $g_i(u_i^*)=1$
that
$$
\sum_{i} \left[ (u_i+v_i)f_i(u_i+v_i)-u_i^* f_i(u_i^*)\right]
\left[g^{-1}_i(u_i+v_i)-g^{-1}_i(u_i^*)\right]\le 0.
$$
As $u_if_i(u_i)$ is strictly increasing and $g_i$ is strictly deceasing,
we have $u_i+v_i=u_i^*$ for every $i$
and the inequality in \eqref{model-2species-6b} must be an equality.
As $(d_{ij})$ is irreducible and $f_i(u_i^*)u_i^*$ is positive
for every $i$, $(d_{ij}f_i(u_i^*)u_i^*)$ is irreducible.
By Theorem \ref{thm:ls}, the equality in \eqref{model-2species-6b} holds
if and only if
\begin{equation}\label{model-2species-9}
\frac{f_i(u_i^*)u_i^*}{f_i(u_i+v_i)u_i}=\frac{f_j(u_j^*)u_j^*}{f_j(u_j+v_j)u_j},
\quad \forall 1\le i, j\le n.
\end{equation}
As $u_i+v_i=u^*_i$ for each $i$, \eqref{model-2species-9}
implies that $u_i^* u_j=u_i u_j^*$ for every $i, j$.
Hence, $u_i=cu_i^*$ for some constant $c>0$.
Since $u_i+v_i=u_i^*$ and $v_i>0$, we have $v_i=(1-c)u_i^*$
for some $c\in (0, 1)$.
Substituting $u_i=cu_i^*$ and $v_i=(1-c)u_i^*$
into \eqref{model-2species-1} and applying $g_i(u_i^*)=1$ we
have
$$
\sum_{j} D_{ij} f_j(u_j^*) u_j^*=f_i(u_i^*) u_i^*
$$
for every $i$, which contradicts the assumption on $(D_{ij})$.
\end{proof}


Next we study the stability of $(0, v^*)$,
where $v^*$ is a componentwise positive solution of
\begin{equation}\label{eq:yi}
s_i(v_i)\sum_{j} D_{ij} f_j(v_j)  v_j=m_i(v_i) v_i, \quad 1\le i\le n.
\end{equation}

\begin{lemma}\label{lemma:c-2} Suppose that {\rm(A1)-(A4)} hold.
If $(d_{ij})$ is an ideal free dispersal strategy with respect to $u^*$ and $(D_{ij})$ is not,
then $(0, v^*)$ is unstable.
\end{lemma}

\begin{proof}
The stability of $(0, v^*)$ is determined by
the dominant eigenvalue, denoted by $\lambda^*$, of the linear problem
\begin{equation}\label{eq:c-20-a}
\lambda \varphi_i=s_i(v_i^*) \sum_{j} d_{ij} f_j(v_j^*)\varphi_j -m_i(v_i^*)\varphi_i, \quad 1\le i\le n.
\end{equation}
Since matrix $(d_{ij})$ is non-negative and irreducible
and $f_j(v_j^*)>0$ for each $j$,
 by the Perron-Frobenius Theorem, the dominant eigenvalue of \eqref{eq:c-20-a} exists,
and the corresponding $\varphi_i$ can be chosen to be positive for every $i$. Multiply \eqref{eq:c-20-a} by $f_i(u_i^*)u_i^*/[s_i(v_i^*) f_i(v_i^*)\varphi_i]$
and sum up the result in $i$.  We have
\begin{equation}
\lambda^* \sum_{i} \frac{ f_i(u_i^*) u_i^*}{s_i(v_i^*) f_i(v_i^*)}
=\sum_{i, j} \red{d}_{ij} \frac{f_j(v_j^*) f_i(u_i^*) u_i^* \varphi_j}{f_i(v_i^*) \varphi_i}
-\sum_{i} g^{-1}_i(v_i^*) f_i(u_i^*) u_i^*, \quad
1\le i\le n.
\end{equation}
Since the matrix $(d_{ij}f_j(u_j^*) u_j^*)$ is \red{ line-sum-symmetric},
\begin{equation}
\aligned
\sum_{i, j}  \red{d}_{ij} \frac{ f_j(v_j^*) \varphi_j f_i(u_i^*) u_i^*}{f_i(v_i^*) \varphi_i}
&=\sum_{i, j}  \red{d}_{ij} f_j(u_j^*) u_j^* \frac{ f_j(v_j^*) \varphi_j f_i(u_i^*) u_i^*}{f_i(v_i^*) \varphi_i f_j(u_j^*) u_j^*}
\\
&\ge \sum_{i, j}  \red{d}_{ij} f_j(u_j^*) u_j^*
\\
&=\sum_{j} f_j(u_j^*) u_j^*,
\endaligned
\end{equation}
where we applied Theorem \ref{thm:ls} by setting
$a_{ij}= \red{d}_{ij} f_j(u_j^*) u_j^*$
and $x_i=f_i(u_i^*) u_i^*/[f_i(v_i^*)\varphi_i]$.

Therefore,
\begin{equation}\label{eq:c-30-a}
\lambda^* \sum_{i} \frac{f_i(u_i^*) u_i^*}{s_i(v_i^*) f_i(v_i^*)}
\ge \sum_{i} f_i(u_i^*) u_i^* [1-g^{-1}(v_i^*)].
\end{equation}

Recall that $v_i^*$ satisfies
\begin{equation}\label{eq:yi*}
s_i(v_i^*) \sum_{j} D_{ij}  f_j(v_j^*) v_j^*-m_i(v_i^*) v_i^*=0, \quad 1\le i\le n.
\end{equation}

Summing the equation of $v_i^*$ over $i$, we have
\begin{equation}\label{eq:c-31-a}
\sum_{i} f_i(v_i^*) v_i^* \left[1-g^{-1}_i(v_i^*)\right]=0.
\end{equation}

Hence, by \eqref{eq:c-30-a} and \eqref{eq:c-31-a}
we have
\begin{equation}\label{eq:c-32-a}
\aligned
\lambda^* \sum_{i} \frac{ f_i(u_i^*) u_i^*}{s_i(v_i^*) f_i(v_i^*)}
&\ge \sum_{i} [1-g^{-1}_i(v_i^*)] \left[f_i(u_i^*) u_i^*-f_i(v_i^*) v_i^*\right]
\\
&=\sum_{i} [g^{-1}_i(u_i^*)-g^{-1}_i(v_i^*)]  \left[f_i(u_i^*) u_i^*-f_i(v_i^*) v_i^*\right]
\\
&\ge 0,
\endaligned
\end{equation}
as both $g^{-1}_i$ and $u_i f_i(u_i)$ are monotone increasing and $g_i(u_i^*)=1$.
It suffices to show that the last inequality of \eqref{eq:c-32-a} is strict.
If not, then $u_i^*=v_i^*$ for every $i$.
By \eqref{eq:yi*} and $g_i(u_i^*)=1$ for every $i$ we see that
$\sum_{j} D_{ij} f_j(u_j^*)u_j^*=f_i(u_i^*) u_i^*$ for every $i$,  which contradicts our assumption.
Therefore, $\lambda^*>0$, i.e. $(0, v^*)$ is unstable.
\end{proof}

\noindent{\bf Proof of Theorem \ref{main:1}.}
By (A1) and (A2),  \eqref{model-2species}
is a strongly monotone dynamical system.
Theorem \ref{main:1} follows from Lemmas \ref{lemma:c-1} and \ref{lemma:c-2}
and Theorem A of \cite{Hsu}.
\qed

\section{Proof of Theorem \ref{main:1B}\label{proof:main:1B}}
{ Let $(\tilde u,\tilde v)$ be a non-zero equilibrium for \eqref{model-2species}. Define $\tilde w=\tilde u+\tilde v$. Assume $\tilde u\neq 0$ (a parallel argument applies if $\tilde v\neq 0$). The proof of Lemma~\ref{lemma:c-1}  implies that $\tilde w= u^*$. }Therefore, $\tilde u$ must be a positive multiple of $u^*$; say $\tilde u =\alpha u^*$ with $\alpha>0$. As $\tilde v = \tilde w- \tilde u=(1-\alpha^*)u^*\ge 0$, $\alpha$ must be $\le 1$. This completes the proof of the first claim.

{To prove the second claim, we will show that $\mathcal{E}$ is a normally hyperbolic attractor in the sense of \citet{hirsch-etal-77}. We present the proof of this claim in the discrete-time case. The proof for the continuous-time case is similar. Define $F,G:\bR_+^n\times \bR_+^n \to \bR_+^n$ by
\[
\begin{aligned}
F_i(u,v)&=&s_i(u_i+v_i) \sum_{j=1}^n d_{ij} f_j(u_j+v_j) u_j
+(1-m_i(u_i+v_i)) u_i\\
 G_i(u,v)&=&s_i(u_i+v_i) \sum_{j=1}^n D_{ij} f_j(u_j+v_j) v_j
+(1-m_i(v_i+v_i)) v_i.
\end{aligned}
\]
The discrete-time dynamics of  \eqref{model-2species} are given by iterating the map $H=(F,G)$. The derivative matrix of $H$ is of the form
\[
J=\begin{pmatrix}
\partial_u F & \partial _v F\\
\partial_u G& \partial_v G
\end{pmatrix}
\]
where $\partial_u F$, $\partial_v F$, $\partial_u G$, and $\partial_v G$ denote the $n\times n$ matrices of partials $\frac{\partial F_i}{\partial u_j}$,  $\frac{\partial F_i}{\partial v_j}$,  $\frac{\partial G_i}{\partial u_j}$, and  $\frac{\partial G_i}{\partial v_j}$, respectively. Assumption (A3) implies that the off-diagonal elements of $\partial_u F$ satisfy
\[
\frac{\partial F_i}{ \partial u_j} = s_i(u_i+v_i) d_{ij} (f_j'(u_j+v_j)u_j+f_j(u_j+v_j))\ge 0 \mbox{ for }j\neq i
\]
with equality if and only if $d_{ij}=0$.
The on-diagonal terms, by assumption, satisfy
\[
\frac{\partial F_i}{ \partial u_i} = \frac{\partial}{\partial u_i} \left( s_i(u_i+v_i)d_{ii}f_i(u_i+v_i)u_i+(1-m_i(u_i+v_i))u_i\right)>0.
\]
By assumption, the entries of $\partial_v F$ satisfy
\[
\frac{\partial F_i}{ \partial v_j} = s_i(u_i+v_i) d_{ij} f_j'(u_j+v_j)u_j\le 0 \mbox{ for }j\neq i
\]
with equality if and only if $d_{ij}=0$ or $u_j=0$,  and
\[
\frac{\partial F_i}{ \partial v_i} = s_i'(u_i+v_i) \sum_{j=1}^n d_{ij} f_j(u_j+v_j) u_j+s_i(u_i+v_i) d_{ii} f_i'(u_i+v_i) u_i
-m_i'(u_i+v_i) u_i\le 0.
\]
Analogous statements apply for $\partial_u G$ and $\partial_v G$ with the $d$ matrix being replaced by the $D$ matrix. }

{For points in $\mathcal{E}^\circ= \{(\alpha u^*,(1-\alpha) u^*): \alpha\in (0,1)\}$, $J$ is primitive with respect to the  competitive ordering  $\ge_K $ on $\bR_+^n\times \bR_+^n$ i.e. $(\tilde u, \tilde v)\ge_K (u,v)$ if $\tilde u_i \ge u_i$ and $\tilde v_i\le v_i$ for all $i$. Since $\mathcal{E}$ is a line of equilibria, $J$ for any point on $\mathcal{E}$ has an
eigenvalue of one associated with the eigenvector $(u^*, -u^*)$. The Perron Frobenius theorem (with respect to the competitive ordering) implies that all the other eigenvalues of
$J$ for points on $\mathcal{E}^\circ$ are strictly less than one in absolute value. }

{ Next consider a point on $\mathcal{E}\setminus \mathcal{E}^\circ=\{ (0,u^*),(u^*,0)\}$, say $(0,u^*)$. At this point, $J$ has a lower triangular block structure as $\partial_v F$ is the zero matrix. Hence, the eigenvalues of $J$ are determined by the matrices $\partial_u F$ and $\partial_v G$ evaluated at $(0,u^*)$. We claim that the dominant eigenvalue of $\partial_v G$ is strictly less than one in absolute value. To see why, we can express the single strategy mapping $v\mapsto G(0,v)$ in the form $v\mapsto A(v)v$ and \[\partial_v G(0,v)=A(v)+\sum_{i=1}^n \frac{\partial A}{\partial v_i}(v)\rm{diag}(v) \] where $\rm{diag}(v)$ denotes a diagonal matrix with diagonal entries $v_1,\dots,v_n$. By assumption, the matrices of partial derivatives $\frac{\partial A}{\partial v_i}$ have entries that are non-positive with some strictly negative. At the equilibrium $v=u^*$ for the single strategy dynamics, we have $A(u^*)=u^*$ and, consequently, the dominant eigenvalue of $A(u^*)$ equals one. As $\partial_v G(0,u^*)$ is a primitive matrix with some entries strictly smaller than the entries of $A(u^*)$ and none of the entries larger, the dominant eigenvalue of $\partial_v G(0,u^*)$ is strictly less than one, as claimed. On the other hand, the matrix $\partial F_u$ evaluated at $(0,u^*)$ is also primitive. Due to the line of equilibria $\mathcal{E}$, this primitive matrix  has a dominant eigenvalue of $1$ and the remaining eigenvalues are strictly less than one in absolute value.}

{Hence, we have shown that  $\mathcal{E}$ is a normally hyperbolic one dimensional attractor. Theorem 4.1 of \cite{hirsch-etal-77} implies that there exists a neighborhood $U\subset \bR_+^n\times \bR_+^n$ of $\mathcal{E}$ and a homeomorphism $h:[0,1] \times V \to U$ with $V=\{z\in \bR^{2n-1}: \|z\|< 1\}$ such that $h(\alpha,0)=(\alpha u^*,(1-\alpha) u^*)$, $h(0,V)=\{(0,v)\in U\}$, $h(1,V)=\{(u,0)\in U\}$, and $\lim_{n\to\infty} H^n(u,v)=(\alpha u^* ,(1-\alpha u^*))$ for all $(u,v)\in h(\{\alpha\}\times V)$.}

\section{Proof of Theorem \ref{main:2}\label{proof:main:2}}
 For $1\le i\le k$, let $u_i^*>0$  be the unique solution to $g_i(u_i^*)=1$. Define $u^* = (u_1^*,\dots,u_k^*,0,0,\dots,0)$. Provided it exists, let $v^*$ be the unique, positive equilibrium to \eqref{model-single} for species $v$, else let $v^*=(0,0,\dots,0)$ be the zero equilibrium. To prove the theorem, we prove two lemmas which imply there are no strongly positive equilibria $(\tilde u,\tilde v)$ and all equilibria, except $(u^*,0)$, are linearly unstable. From the theory of monotone dynamical systems~\citep{Smith} it follows that all solutions with strongly positive initial conditions converge to $(u^*,0)$. \medskip

\begin{lemma}~\label{char}  Let $(\tilde u ,\tilde v)$ be a component-wise non-negative equilibrium of \eqref{model-2species-v2}. Then \red{one of the following statements holds}: (i) $(\tilde u,\tilde v)=(0,0)$, (ii) $\tilde u_i =u_i^*$ for some $1\le i\le k$ and $\tilde v=0$, (iii)  $(\tilde u,\tilde v)=(0,v^*)$, and (iv) $\tilde u_i >0$ for some $1\le i \le k$, $\tilde v_j>0$ for some $j$, and $\tilde u_\ell =0$ and $g_\ell(\tilde v_\ell)>1$ for some $1\le \ell \le k$.\end{lemma}

\begin{proof}
(i)-(iii) describe all equilibria $(\tilde u, \tilde v)$ where either $\tilde u=0$ or $\tilde v=0$. Consider an equilibrium where $\tilde u\neq 0$ and $\tilde v\neq 0$. Then $\tilde u_i>0$ for some $1\le i \le k$, and $\tilde v_j>0$ for some $j$. In fact, irreducibility of $D$ implies $\tilde v_i>0$ for all $i$. As $g_i$ are decreasing functions and $g_i(0)\le 1$ for $i>k$, $\tilde u_i=0$ for $i>k$. Define $S$, $\Lambda$, $M$ to be the diagonal matrices with diagonal entries $s_i(\tilde u_i+\tilde v_i)$, $f_i(\tilde u_i+\tilde v_i)$, $m_i(\tilde u_i+\tilde v_i)$. The equilibrium condition $\Delta v=0$ implies that $0=\rho(SD\Lambda -M)$. Lemma~\ref{reduction} implies
\[
0=\rho(SD\Lambda -M)< \rho(S\Lambda-M)=\max_i \{s_i(\tilde u_i+\tilde v_i)f(\tilde u_i+\tilde v_i)-m_i(\tilde u_i+\tilde v_i)\}.
\]
Hence, there exists some $1\le i\le k$ such that $g_i(\tilde v_i)\ge g_i(\tilde u_i +\tilde v_i)>1$.
\end{proof}

\begin{lemma} The equilibrium $(u^*,0)$ is linearly stable and all other equilibria $(\tilde u, \tilde v)$ are unstable. \end{lemma}

\begin{proof} We begin by showing $(u^*,0)$ is linearly stable.  Let $S$, $\Lambda$, and $M$  be the diagonal matrices with diagonal entries $s_i(u_i^*)$, $f_i(u_i^*)$, and $m_i(u_i^*)$. We have $\rho(S\Lambda-M)=\max_i s_i(u_i^*)f(u_i^*)-m_i(u_i^*) = 0.$ {As $g_i(0)<1$ for some $i>k$}, Lemma~\ref{reduction} implies  $\rho(SD\Lambda -M)< \rho(S\Lambda-M)=0$. Hence, $(u^*,0)$ is linearly stable.

To show the remaining types of equilibria are unstable, we consider the four cases given by Lemma~\ref{char}. For cases (i), (ii) with $\tilde u \neq u^*$, and (iv),  $\tilde u_i=0$ for some $1\le i \le k$ and, consequently, $s_i(\tilde u_i ) f_i(\tilde u_i)-m_i(\tilde u_i)>0$ for this $i$, ensuring instability. For case (iii), define $S$, $\Lambda$, and $M$ to be the diagonal matrices with  diagonal entries $s_i(v_i^*)$, $f_i(v_i^*)$, and $m_i(v_i^*)$. Then  $0=\rho(SD\Lambda -M)$. Lemma~\ref{reduction} implies  $0=\rho(SD\Lambda -M)< \rho(S\Lambda-M)$. Hence, $(0,v^*)$ is linearly unstable.
\end{proof}

\begin{lemma} If $u_0$ and $v_0$ are componentwise positive, then the $\omega$-limit set $\omega((u_0,v_0))=\{(u^*,0)\}$. \end{lemma}

\begin{proof}
	Consider $(\bar{u},\bar{v})$, where $(\bar{u},\bar{v})$ is a componentwise nonnegative equilibrium of (3.9). Note that if for some $i, \bar u_i\not= 0 $, then $g_i(\bar{u}_i+\bar{v}_i)=1$, which implies that $g_i(0)>1$.
	
	Consequently, if $g_i(\bar{u}_i+\bar{v}_i)\not=1$, $\bar{u}_i=0$. In particular, whenever $g_i(0)\le1$ and $\bar{v_i}\not=0$, $\bar u_i=0$. Note that $\bar v$ satisfies
	\begin{equation*}
	\bar v_i=\frac{s_i(\bar u_i+\bar v_i)}{m_i(\bar u_i+\bar v_i)}\sum^{n}_{j=1}D_{ij}f_j(\bar u_j+\bar v_j)\bar v_j,\quad 1\le i\le n.
	\end{equation*}
Since the matrix $ D=(D_{ij})$ is irreducible, either $\bar v\equiv 0$ or $\bar v_i>0$ for $i=1,...,n$. Suppose $\bar{v}\not\equiv0$. Let $S$, $\Lambda$ and $M$ be as in the proof of Lemma 6.1. By Lemma 2.2,
\[0=\rho(SD\Lambda-M)\le\rho(S\Lambda-M).\]
Consequently, there is an $i$ so that
\[s_i(\bar u_i+\bar v_i) f_i(\bar u_i+\bar v_i)-m(\bar u_i+\bar v_i)\ge0,\]
which implies that $g_i(\bar u_i+\bar v_i)\ge1$. Since $g_i(0)\le1$ for $j>k$, $s_j(\bar u_j+\bar v_j) f_j(\bar u_j+\bar v_j)-m(\bar u_j+\bar v_j)<0$ for $j>k$. So $S\Lambda-M$ is not a scalar matrix. Hence $\rho(SD\Lambda-M)<\rho(S\Lambda-M)$.

So there is an $i$ so that
\[s_i(\bar u_i+\bar v_i) f_i(\bar u_i+\bar v_i)-m(\bar u_i+\bar v_i)>0,\]
and thus
$g_i(\bar u_i+\bar v_i)>1$. Since $g_i(\bar u_i+\bar v_i)\not=1$, $\bar u_i=0$ and $g_j(0)\le1$ for $j>k$, $i\le k$.

Hence if $\bar v\not\equiv0$, there is an $i\in\{1,...,k\}$ so that $\bar{u}_i=0$ and $g_i(\bar{v}_i)>1$. Note that $g_i(\bar v_i)>1$ implies
\[s_i(\bar v_i)f_i(\bar v_i)-m_i(\bar v_i)>0.\]
Now let $(u(t),v(t))$ be any trajectory (in either discrete or continuous time) corresponding to $(u(0),v(0))=(u_0,v_0)$. (Here $u_0$ and $v_0$ are componentwise positive.) If there is a sequence of times $t_n\to \infty$ as $n\to \infty$ so that
\[(u(t_n),v(t_n))\to(\bar u,\bar v),
\]
then there is an $N$ so that for $n\ge N$
\begin{align*}
\frac{(\Delta u)_i(t_n)}{u_i(t_n)}&= s_i(u_i(t_n)+v_i(t_n))f_i(u_i(t_n)+v_i(t_n))-m_i(u_i(t_n)+v_i(t_n))\\
&>\frac{s_i(\bar v_i)f_i(\bar v_i)-m_i(\bar v_i)}{2}.
\end{align*}
Consequently, $u_i(t_n)\not\to0$. So there can be no $(\bar u,\bar v)$ with $\bar v\not\equiv 0$ for which $(\bar u,\bar v)\in \omega ((u_0,v_0))$. So it must be the case that $\bar v\equiv 0$ if $(\bar u,\bar v)\in \omega ((u_0,v_0))$.

Suppose there is a sequence of times $t_n\to \infty$ as $n\to \infty$ so that
\[(u(t_n),v(t_n))\to(\bar u,0).
\]
For $i\in\{1,...,k\}$, we have $\bar u_i=0$ or $\bar u_i=u^*_i$. If $\bar u_i=0$, there is an
 $N$ so that for $n\ge N$
\begin{align*}
\frac{(\Delta u)_i(t_n)}{u_i(t_n)}&= s_i(u_i(t_n)+v_i(t_n))f_i(u_i(t_n)+v_i(t_n))-m_i(u_i(t_n)+v_i(t_n))\\
&>\frac{s_i(\bar v_i)f_i(\bar v_i)-m_i(\bar v_i)}{2}\\
&>0.
\end{align*}
Since $g(u^*_i)=1$ implies $g_i(0)>1$, the only possibility is that $\bar u=u^*$. The upshot is that $(u^*,0)$ is globally attracting relative to initial data which is componentwise positive.

\end{proof}

\end{document}